\documentclass{article}

\usepackage{amsmath,amssymb}

\newtheorem{thm}{Theorem}

\newtheorem{defn}{Definition}

\newtheorem{open}{Open Problem}

\newcommand{\tfr}[1]{{\widehat{#1}}}
\newcommand{\F}[1]{\mathbb{F}_{2^{#1}}}

\newenvironment{proof}
{\noindent {\bf Proof.}} {\hfill $\Box$\\}

\DeclareMathOperator{\NL}{NL}
\DeclareMathOperator{\Tr}{Tr}
\begin{document}

\title{A Highly Nonlinear Differentially 4 Uniform Power Mapping That Permutes Fields of Even Degree}

\author{Carl Bracken$^1$,  Gregor Leander$^2$\\
$^1$Department of Mathematics, National University of Ireland\\
Maynooth, Co. Kildare, Ireland\\
$^2$Department of Mathematics, Technical University Denmark\\
Copenhagen, Denmark}

\maketitle

\begin{abstract}
Functions with low differential uniformity can be used as the s-boxes of symmetric cryptosystems as they have good resistance to differential attacks. The AES (Advanced Encryption Standard) uses a differentially-4 uniform function called the inverse function. Any function used in a symmetric cryptosystem should be a permutation. Also, it is required that the function is highly nonlinear so that it is resistant to Matsui's linear attack. In this article we demonstrate that the highly nonlinear permutation $f(x)=x^{2^{2k}+2^k+1}$, discovered by Hans Dobbertin \cite{Dob}, has differential uniformity of four and hence, with respect to differential and linear cryptanalysis, is just as suitable for use in a symmetric cryptosystem as the inverse function.
\end{abstract}

\maketitle

\section{Introduction}
Functions with a low differential uniformity are interesting from the point of view of cryptography as they provide good resistance to differential attacks \cite{N}. For a function to be used as an s-box of a symmetric cryptosystem it should be a permutation and defined on a field with even degree. It is also essential that the function has high nonlinearity so that it is resistant to Matsui's linear attack \cite{Mat}. The lowest possible differential uniformity is 2 and functions with this property are called APN (almost perfect nonlinear). There has been much recent work and progress on APN functions (see  \cite{BB},\cite{BB1},\cite{BCL},\cite{BCL1},\cite{BCL2}). However, at present there are no known APN permutations defined on fields of even degree and it is actually the most important open question in this field if such functions exist. This is why the AES (advanced encryption standard) uses a differentially 4 uniform function, namely the inverse function.

For the rest of the paper, let $L = \F{n}$ for $n > 0$ and let $L^*$ denote the set of non-zero elements of $L$. Let $\Tr: L \rightarrow \F{}$ denote the trace map from $L$ to $\F{}$. For positive integers $r,k$ by $\Tr^{rk}_k$ we denote the relative trace map from $\F{rk}$ to $\F{k}$ and by $\Tr^{r}$ the absolute trace from $\F{r}$ to $\F{}$.

\begin{defn}
A function $f: L \rightarrow L^*$ is said to be differentially $\delta$ uniform.
if for any $a\in L^*, b \in L$,  we have
$$ |\{x \in L : f(x+a)+f(x) = b \}| \leq \delta. $$
\end{defn}

\begin{defn}
Given a function $f:L \rightarrow L$, the {\emph {\bf{Fourier transform}}} of $f$ is the function
$\tfr f: L \times L^* \rightarrow \mathbb Z$ given by
$$\tfr f(a,b) = \sum_{x \in L}(-1)^{\Tr(ax+bf(x))}.$$
\label{WT}

\end{defn}

The {\emph{Fourier spectrum}} of $f$ is the set of integers
$$\Lambda_f= \{\tfr f(a,b) : a, b \in L, b \neq 0 \}.$$

The nonlinearity of a function $f$ on a field $L=\F{n}$ is defined as

$$NL(f):= 2^{n-1} - \frac{1}{2}\max_{x \in \Lambda_f} \ |x| .$$

The nonlinearity of a function measures its distance
to the set of all affine maps on $L$.
We thus call a function {\emph{maximally nonlinear}} if its nonlinearity is as large as possible.
If $n$ is odd, its nonlinearity is upper-bounded by $2^{n-1}-2^{\frac{n-1}{2}}$, while for $n$ even a conjectured upper bound is  $2^{n-1}-2^{\frac{n}{2}-1}$.
For odd $n$, we say that a function $f : L\longrightarrow L$ is \emph{almost bent} (AB)
when its Fourier spectrum is $\{0, \pm 2^{\frac{n+1}{2}}\} $,
in which case it is clear from the upper bound that $f$ is maximally nonlinear.

In an article of Hans Dobbertin \cite{Dob} he offers a list of power mappings that permute fields of even degree and meet the conjectured nonlinearity bound of $2^{n-1}-2^{\frac{n}{2}-1}$. Following Dobbertin's terminology we shall refer to such mapping as highly nonlinear permutations.
In \cite{Dob} Dobbertin conjectured that this list was complete and noted that this had been verified for $n \leq 22$. The inverse function (used in AES) is highly nonlinear and hence is on the list. One of the functions on Dobbertin's list is the power mapping
$f(x)=x^{2^{2k}+2^{k}+1}$, defined on $\F{4k}$, with $k$ odd.

In this article we show that this function has differential uniformity of 4. We also provide another proof of this functions nonlinearity property. This means that this function has the same resistance to both the linear and differential attacks as the inverse function.

\section{Differential Uniformity of $f(x)=x^{2^{2k}+2^k+1}$}
As mentioned above there are no known permutations of even degree fields with differential uniformity of two. The following theorem shows that $x^{2^{2k}+2^k+1}$ has the next best
(and best known) differential uniformity, which is four.

\begin{thm}
Let $f(x)=x^{2^{2k}+2^k+1}$ be defined on $\F{4k}$. Then $f(x)$ has differential uniformity of four.
\end{thm}
\begin{proof}
We need to demonstrate that the equation
$$x^{2^{2k}+2^k+1}+(x+a)^{2^{2k}+2^k+1}=b$$
has no more than four solutions for all $a \in {\F{4k}}^*$ and all $b \in \F{4k}$.

Expansion of this expression yields
$$ax^{2^{2k}+2^k}+a^{2^k}x^{2^{2k}+1}+a^{2^{2k}}x^{2^k+1}+a^{2^k+1}x^{2^{2k}}+a^{2^{2k}+1}x^{2^k}+a^{2^{2k}+2^k}x+a^{2^{2k}+2^k+1}=b.$$
Next we replace $x$ with $xa$ and divide by $a^{2^{2k}+2^k+1}$ and obtain
\begin{eqnarray}\label{eqn:1}
x^{2^{2k}+2^k}+x^{2^{2k}+1}+x^{2^k+1}+x^{2^{2k}}+x^{2^k}+x+c&=&0
\end{eqnarray}
where $c=a^{-2^{2k}-2^k-1}b+1$.

Let $\Tr^{4k}_k$ denote the relative trace map from $\F{4k}$ to $\F{k}$.

As $\Tr^{4k}_k (x^{2^{2k}+2^k}+x^{2^{2k}+1}+x^{2^k+1}+x^{2^{2k}}+x^{2^k})=0$, Equation (\ref{eqn:1}) implies
$\Tr^{4k}_k(x+c)=0$.

Which is equivalent to
\begin{eqnarray} \label{eqn:t}
x+x^{2^k}+x^{2^{2k}}+x^{2^{3k}}=t
\end{eqnarray}
where $t=\Tr^{4k}_k (c)$. We note that $t \in \F{k}$.

Equation (\ref{eqn:1}) now becomes
$$x(x^{2^k}+x^{2^{2k}})+x^{2^{2k}+2^{k}}+x^{2^{3k}}+t+c=0.$$
Which implies
$$x(x+x^{2^{3k}}+t)+x^{2^{2k}+2^{k}}+x^{2^{3k}}+t+c=0.$$
From which we obtain
\begin{eqnarray} \label{eqn:2}
 x^2+x^{2^{3k}+1}+xt+x^{2^{2k}+2^{k}}+x^{2^{3k}}+t+c &=& 0.
\end{eqnarray}
We raise Equation (\ref{eqn:2}) by $2^{2k}$ and get
\begin{eqnarray} \label{eqn:3}
x^{2^{2k+1}}+x^{2^{k}+2^{2k}}+x^{2^{2k}}t+x^{2^{3k}+1}+x^{2^{k}}+t+c^{2^{2k}}&=&0.
\end{eqnarray}
Now we add Equations (\ref{eqn:2}) and (\ref{eqn:3}) and make use of (\ref{eqn:t}). This gives
\begin{eqnarray} \label{eqn:4}
 (x+x^{2^{2k}})^2+(t+1)(x+x^{2^{2k}})+c^{2^k}+c^{2^{3k}}&=&0.
\end{eqnarray}
The remainder of the proof is divided into two cases. They are $t=1$ and $t \neq 1$.

If $t=1$ then Equation (\ref{eqn:4}) implies
$$x+x^{2^{2k}}=c^{2^{k-1}}+c^{2^{3k-1}}.$$

We let $r=c^{2^{k-1}}+c^{2^{3k-1}}$. Therefore $x^{2^{2k}}=x+r$. Placing this into Equation (\ref{eqn:1}) yields
$$(x+r)x^{2^k}+(x+r)x+x^{2^{k}+1}+x^{2^k}+c+r=0.$$
Which we write as
\begin{eqnarray} \label{eqn:5}
x^2+r(x+x^{2^k})+x^{2^k}+r+c&=&0.
\end{eqnarray}
Raising Equation (\ref{eqn:5}) by $2^k$ we obtain
\begin{eqnarray} \label{eqn:6}
x^{2^{k+1}}+r^{2^k}(x^{2^k}+x+r)+x+r+r^{2^{k}}+c^{2^k}&=&0.
\end{eqnarray}
Next we add Equations (\ref{eqn:5}) and (\ref{eqn:6}) to get
\begin{eqnarray} \label{eqn:7}
 (x+x^{2^k})^2+(r+r^{2^k}+1)(x+x^{2^k})+r^{2^k+1}+c+c^{2^k}+r^{2^k}&=&0.
\end{eqnarray}

Note that $r+r^{2^k}=x+x^{2^k}+x^{2^{2k}}+x^{2^{3k}}=t$, hence if $t=1$ Equation (\ref{eqn:7}) becomes
$$(x+x^{2^k})^2+r^{2^k+1}+c+c^k+r^{2^k}=0.$$
This implies $x+x^{2^k}=s$ where $s=\sqrt{r^{2^k+1}+c+c^k+r^{2^k}}$. Now we replace $x^{2^k}$ by $x+s$ in Equation (\ref{eqn:5}) and obtain
$$x^2+x+rs+s+r+c=0,$$
which can have no more than two solutions in $x$.

Next we consider the case $t \neq 1$.

We replace $x$ with $(t+1)z$ in Equation (\ref{eqn:4}) and get
$$(t+1)^2((z+z^{2^{2k}})^2+(z+z^{2^{2k}}))+c^{2^k}+c^{2^{3k}}=0.$$
Now let $y=z+z^{2^{2k}}$ so we have
$$(t+1)^{2}(y^2+y)+c^{2^k}+c^{2^{3k}}=0.$$
This equation has at most two solutions in $y$. They are of the form $y=p$ and $y=p+1$ for some fixed $p$.
This implies that $z^{2^{2k}}=z+p$ or $z^{2^{2k}}=z+p+1$. Note that $p \in \F{2k}$.

If $z^{2^{2k}}=z+p$ then Equation (\ref{eqn:1}) becomes
$$(t+1)^{2}((z+p)z^{2^k}+(z+p)z+z^{2^k+1})+(t+1)(z^{2^k}+p)+c=0,$$
which gives
\begin{eqnarray} \label{eqn:8}
(t+1)^{2}((z+z^{2^k})p+z^2)+(t+1)(z^{2^k}+p)+c&=&0.
\end{eqnarray}
We raise Equation (\ref{eqn:8}) by $2^k$ and obtain
\begin{eqnarray} \label{eqn:9}
(t+1)^{2}((z+z^{2^k}+p)p^{2^k}+z^{2^{k+1}})+(t+1)(z+p+p^{2^k})+c^{2^k}&=&0.
\end{eqnarray}
Next we add Equations (\ref{eqn:8}) and (\ref{eqn:9}) to get
$$(t+1)^{2}((p+p^{2^k})(z+z^{2^k})+(z+z^{2^{k}})^2+p^{2^k+1})+(t+1)(z+z^{2^k}+p^{2^k})+c+c^{2^k}=0,$$
which becomes
\begin{eqnarray}
(t+1)^{2}(z+z^{2^k})^2 + ((t+1)^{2}(p+p^{2^k})+(t+1))(z+z^{2^k}) \nonumber \\
 +(t+1)^{2}p^{2^k+1}+(t+1)p^{2^{k}}+c+c^{2^k}&=&0.  \label{eqn:10}
\end{eqnarray}

Recall $t=x+x^{2^k}+x^{2^{2k}}+x^{2^{3k}}=(t+1)(z+z^{2^k}+z^{2^{2k}}+z^{2^{3k}})$.

Also $p+p^{2^k}=z+z^{2^k}+z^{2^{2k}}+z^{2^{3k}}$, hence $p+p^{2^k}=\frac{t}{t+1}$.

Therefore Equation (\ref{eqn:10}) becomes
\begin{eqnarray}
(t+1)^{2}((z+z^{2^k})^2+(z+z^{2^k}))+(t+1)^{2}p^{2^k+1} \nonumber \\+(t+1)p^{2^{k}}+c+c^{2^k}&=&0. \label{eqn:11}
\end{eqnarray}

It can easily be verified that if we had assumed $z^{2^{2k}}=z+p+1$ then the same computations as above would also yield Equation (\ref{eqn:11}), so this case need not be considered.

Next we let $z+z^{2^k}=w$ and write Equation (\ref{eqn:11}) as
$$(t+1)^{2}(w^2+w)+(t+1)^{2}p^{2^k+1}+(t+1)p^{2^{k}}+c+c^{2^k}=0.$$
This equation has at most two solutions in $w$ which take the form $w=q$ and $w=q+1$ for some fixed $q$.

This implies that $z^{2^k}=z+q$ or $z^{2^{k}}=z+q+1$.

If $z^{2^k}=z+q$ then $z^{2^{2k}}=z+q+q^{2^k}$ and Equation (\ref{eqn:1}) becomes
$$(t+1)^{2}((z+q+q^{2^k})(z+q)+(z+q+q^{2^k})z+(z+q)z)+(t+1)(z+q^{2^k})+c=0.$$
This simplifies to
$$(t+1)^{2}z^2+(t+1)z+(t+1)^2(q^{2^k+1}+q^2)+(t+1)q^{2^k}+c=0,$$
which is the same as
$$x^2+x+(t+1)^2(q^{2^k+1}+q^2)+(t+1)q^{2^k}+c=0.$$

If on the other hand  $z^{2^k}=z+q+1$, then we would obtain
$$x^2+x+(t+1)^2(q^{2^k+1}+q^{2^k}+q^2+q)+(t+1)(q+1)^{2^k}+c=0.$$
Clearly, this pair of equations will allow no more than four solutions in $x$ and the proof is complete.
\end{proof}

Note that we did not need to assume that $k$ is odd to derive the differential uniformity of four, however it is easy to see that the function is not a permutation if $k$ is even as $g.c.d.(2^{4k}-1, {2^{2k}+2^{k}+1})=1$ if and only if $k$ is odd.
\bigskip

\section{Nonlinearity of $f(x)=x^{2^{2k}+2^k+1}$ }
In this section we give a slightly different proof of the fact that $x^{2^{2k}+2^{k}+1}$ has  $\NL(f)=2^{n-1}-2^{\frac{n}{2}-1}$. Most importantly, our proof also covers the case where the function is not a permutation, i.e., when $k$ is even.

Technically, the main difference to Dobbertin's proof in \cite{Dob} is that we are not going to use an $\F{k}$ basis of $\F{4k}$ to express elements in $\F{4k}$ but rather a $\F{2k}$ basis. This change makes some of the ``lengthy but routine'' computations, as Dobbertin states it, easier.
\begin{thm}
Let $f(x)=x^{2^{2k}+2^k+1}$ be defined on $\F{4k}$. Then
\[ \NL(f)=2^{n-1}-2^{\frac{n}{2}-1}. \]
\end{thm}
\begin{proof}
We have to show that for any non-zero $b$ and any $a$ the absolute value of the Fourier coefficient $\tfr{f}(a,b) $ is smaller or equal to $2^{2k+1}$. There are two cases to consider. If $k$ is odd, then $f$ is a bijection and it is therefore enough to study the case $b=1$. If $k$ is even, then $\gcd(2^{2k}+2^k+1,2^{4k}-1)=3$ and up to equivalence there are two different $b$ to consider, namely the case $b=1$ and $b$ any non-cube. Here we remark that in the case $k$ even we can always choose a non cube in $\F{k}$ with out loss of generality. Thus, in both cases it is enough to study $b\in \F{r}$. Moreover, we can restrict the case to elements $b \in \F{k}$ such that $\Tr^{{k}}(b)=1$.

Let $\gamma$ be any non-zero element in $\F{k}$ such that $\Tr^{{k}}(\gamma)=1$. For simplicity we denote by $g_{\gamma^2}(x)=\Tr(\gamma^2 x^{2^{2k}+2^k+1})$ (we use $\gamma^2$ instead of $\gamma$ to avoid dealing with square roots later on) . Furthermore, let $\alpha\in \F{2k}$ be an element fulfilling the equation $\alpha^2+\gamma \alpha+\gamma^3=0$. As $\Tr^{k}(\gamma)=1$ the polynomial  $\alpha^2+\alpha+\gamma=0$ is irreducible over $\F{k}$ and by replacing $\alpha$ by $\alpha \gamma^{-1}$ and multiplying across by $\gamma^2$ we see that the polynomial  $\alpha^2+\gamma \alpha+\gamma^3=0$ is irreducible as well. Therefore $\alpha \notin \F{k}$ and furthermore it holds that $\alpha^{2^k}+\alpha=\gamma$. Thus,
\[ \Tr^{{2k}}(\alpha)=\Tr^{k}(\alpha^{2^k}+\alpha)=\Tr^{{k}}(\gamma)=1. \]
This implies that the polynomial $x^2+x+\alpha$ is irreducible over $\F{2k}$ and finally every element in $\F{4k}$ can be represented by $y+\omega a$, where $y,a \in \F{2k}$ and $\omega \in \F{4k}$ with $\omega^2+\omega+\alpha=0$. Using this expression for $x$ we compute
\begin{eqnarray*}
g_{\gamma^2}(x)&=& g_{\gamma^2}(y+\omega a) \\
&=& \Tr( \gamma^2(y+\omega a)^{2^{2k}+2^k+1}) \\
&=& \Tr( \gamma^2 y^{2^{2k}+2^k+1})\\
&&+\Tr(\gamma^2( y^{2^{2k}+2^k}(\omega a)+y^{2^{2k}+1}(\omega a)^{2^k}+y^{2^k+1}(\omega a)^{2^{2k}})) \\
&&+\Tr(\gamma^2 ( y^{2^{2k}}(\omega a)^{2^k+1}+y^{2^k}(\omega a)^{2^{2k}+1}+y(\omega a)^{2^{2k}+2^k} ))\\
&&+\Tr( \gamma^2 (\omega a)^{2^{2k}+2^k+1}) \\
&=& A+B+C+D.
\end{eqnarray*}
First we note that $A=0$ as $\gamma^2$ and $y$ are in $\F{{2k}}$. Furthermore $B$ can be simplified,
\begin{eqnarray*}
 B&=&\Tr(\gamma^2 y^{2^k+1}(\omega a)+\gamma^2 y^2(\omega a)^{2^k}+\gamma^2 y^{2^k+1}(\omega a)^{2^{2k}}) \\
 &=&\Tr( \gamma^2 y^{2^k+1}( (\omega a)+(\omega a)^{2^{2k}})+\gamma^2 y^2(\omega a)^{2^k}) \\
 &=& \Tr(\gamma^2 y^2(\omega a)^{2^k}).
 \end{eqnarray*}
where the last equality follows as $\gamma^2 y^{2^k+1}( (\omega a)+(\omega a)^{2^{2k}})$ is in $\F{2k}$. Now consider the term $C$. We first remark that $\gamma^2 y^{2^k}(\omega a)^{2^{2k}+1}$ is in the subfield $\F{2k}$ and thus
\begin{eqnarray*}
C&=&+\Tr(\gamma^2 (y^{2^{2k}}(\omega a)^{2^k+1}+y(\omega a)^{2^{2k}+2^k} ))\\
&=& \Tr(\gamma^2 y( (\omega a)^{2^k+1}+(\omega a)^{2^{2k}+2^k})).
\end{eqnarray*}
Therefore
\begin{eqnarray*}
g(x)&=& g(y+\omega a) \\
&=& \Tr( y\left( \gamma (\omega a)^{2^{k-1}}+\gamma^2(\omega a)^{2^k+1}+\gamma^2(\omega a)^{2^{2k}+2^k}\right)+  \gamma^2(\omega a)^{2^{2k}+2^k+1}).
\end{eqnarray*}
The important observation is that this expression is linear in $y$. Thus, the function belongs to the generalized Maiorana McFarland type of functions. Next, we compute an expression of $g$ using the absolute trace on $\F{2k}$ denoted by $\Tr^{2k}$. For this we make use of the following equations
\[ \omega+\omega^{2^{2k}}=1\]
and
\[ \omega^{2^{2k}+2^k+1}+(\omega^{2^{2k}+2^k+1})^{2^{2k}}=\alpha \]
that follow from the fact that the two solutions of
\[ x^2+x+\alpha =0 \]
are $\omega$ and $\omega^{2^{2k}}$.
\begin{eqnarray*}
g(y+\omega a)&=& \Tr^{2k}( y \left( \gamma a^{2^{k-1}}(\omega+\omega^{2^{2k}})^{2^{k-1}}+ \gamma^2(a(\omega+\omega^{2^{2k}}))^{2^k+1} \right))\\
&&+ \Tr^{2k}( \gamma^2 a^{2^{2k}+2^k+1}( \omega^{2^{2k}+2^k+1}+(\omega^{2^{2k}+2^k+1})^{2^{2k}})) \\
&=& \Tr^{2k}( y \left( \gamma a^{2^{k-1}}+\gamma^2 a^{2^k+1} \right) +\alpha \gamma^2 a^{2^k+2}).
\end{eqnarray*}
From now on the proof continues very much like Dobbertin's original proof. We denote by $\mu(y)=(-1)^{\Tr^{2k}(y)}$ and
\[ \pi(a)=\gamma a^{2^{k-1}}+\gamma^2a^{2^k+1}.\]
We compute
\begin{eqnarray*}
\tfr{g}(u+\omega v)&=& \sum_{y,a \in \F{2k} }\mu( y\pi(a)+ \alpha \gamma^2a^{2^k+2}+uy+va) \\
&=&  \sum_{a} \mu( \alpha \gamma^2 a^{2^k+2}+va) \sum_y \mu( y(\pi(a)+u)) \\
&=& 2^{2k} \sum_{a, \pi(a)=u} \mu( \alpha \gamma^2 a^{2^k+2}+va).
\end{eqnarray*}
For any $u$ we have to study the set $M=\{a \ | \ \pi(a)=u\}$ and in particular its possible size. First note that
\[ \pi(a)=\pi(a+c) \]
implies
\begin{eqnarray*}
0 &=& \pi(a)+\pi(a+c)+(\pi(a)+\pi(a+c))^{2^k} \\
&=& \gamma(c^{2^{k-1}}+c^{2^{2k-1}}) \\
&=& \gamma(c+c^{2^k})^{2^{k-1}}
\end{eqnarray*}
and we conclude $c \in \F{k}$.  Therefore, we can equivalently study the set
\[ \{c^2  \in \F{k} \ | \ \pi(a_0+c^2)=u\} \]
where $a_0$ is an element in $M$. Note that we use $c^2$ instead of $c$ to get rid of the power $2^{k-1}$.
Considering the equation
\[ \pi(a_0)+\pi(a_0+c^2) =0\]
we get the following equation for $c$
\begin{eqnarray}\label{eqn:c}
c^4+(a_0^{2^k}+a_0)c^2+\gamma^{-1}c&=&0
\end{eqnarray}
which immediately implies $|M| \in \{0,1,2,4\}$. As $\tfr(g)(u+\omega v) \le 2^{2k}|M|$ the only case we need to care about for proving the theorem is the case $|M|=4$. In this case the set $M$ consists of elements
\[ M=\{ a_0,a_0+c_0,a_0+c_1,a_0+c_0+c_1 \} \]
where $c_0,c_1$ are solutions of (\ref{eqn:c}) and thus $c_0c_1(c_0+c_1)=\gamma^{-1}$. Next we compute
\begin{eqnarray*}
\sum_{a \in M} \Tr^{2k}( \alpha \gamma^2 a^{2^k+2}+va) &=& \Tr^{2k}( \alpha \gamma^2 ( a_0^{2^k+2}+(a_0+c_0)^{2^k+2}\\
& & +(a_0+c_1)^{2^k+2}+(a_0+c_0+c_1)^{2^k+2}))\\
& & \\
&=& \Tr^{2k}(\alpha \gamma^2( c_0c_1^2+c_1c_0^2)) \\
&=& \Tr^{2k}(\alpha \gamma^2( c_0c_1(c_0+c_1))) \\
&=& \Tr^{2k}(\alpha \gamma) \\
&=& \Tr^{k}(\gamma (\alpha+\alpha^{2^k})) \\
&=& \Tr^k(\gamma^2)=1
\end{eqnarray*}
which implies
\[\tfr{g}(u+\omega v)= \sum_{a \in M} \mu( \alpha a^{2^k+2}+va) = \pm 2^{2k+1} .\]
\end{proof}

\section{Closing Remarks and Open Problems}
We have demonstrated that the function $f(x)=x^{2^{2k}+2^{k}+1}$ has the same resistance to both differential and linear attacks as the inverse function. The fact that it can permute the field when $k$ is odd means it could be used in a cryptosystem acting on 12 bits.
We now list all the known highly nonlinear permutations with differential uniformity of 4.
For power mappings we conjecture this list to be complete.
\begin{center}

 \begin{tabular}{|c|c|c|}
 \hline
  &  &  \\
 ${\bf f(x)} $ & {\bf Conditions} & {\bf Ref. }\\
 \hline
  &  &  \\

 $x^{2^s+1}$  & $n=2k$, \ $k$ odd   & \cite{G}  \\
 & $gcd(n,s)=2$ &  \\
 \hline
   &  & \\
 $x^{2^{2s}-2^s+1}$ &  $n=2k$, \ $k$ odd   & \cite{K}  \\
 &  $gcd(n,s)=2$  &  \\
 \hline
  &  &  \\
$x^{-1}$  & $n$ even  & \cite{BD}\\
  &  &  \\
 \hline
  &  &  \\
  $x^{2^{2k}+2^{k}+1}$ &  $n=4k$, \ $k$ odd &  This article\\
  &  &  \\

 \hline
 \end{tabular}.
\end{center}

\begin{open}
Find more highly nonlinear permutations of even degree fields with differential uniformity of 4.
\end{open}

\begin{open}
Find a function, defined on a field of even degree, with higher nonlinearity than $2^{n-1}-2^{\frac{n}{2}-1}$ or prove that such a function can't exist.
\end{open}

\end{document}